\documentclass[a4paper,aps,pra,showpacs,superscriptaddress,twocolumn,nofootinbib]{revtex4-1}      
\usepackage[utf8x]{inputenc}  
\usepackage[T1]{fontenc}     
\usepackage[british]{babel}  
\usepackage[sc,osf]{mathpazo}
\usepackage[scaled=0.86]{berasans}  
\usepackage[scaled=1.03]{inconsolata} 

\usepackage{color}
\usepackage[usenames,dvipsnames]{xcolor}
\definecolor{myrefcolor}{rgb}{0.,0.6,0.}
\usepackage[colorlinks,linkcolor=magenta,urlcolor=blue,citecolor=green]{hyperref}  

\usepackage{graphicx} 
\usepackage[babel]{microtype}  
\usepackage{amsmath,amssymb,amsthm,bm,amsfonts,mathrsfs,bbm} 
\usepackage{xspace}  
\usepackage{centernot}
\usepackage{pgfplots}  
\usepackage{changes}

\DeclareMathOperator{\tr}{tr}

\newcommand{\proj}[1]{\ket{#1}\!\bra{#1}}
\newcommand{\bra}[1]{\left\langle #1 \right|}
\newcommand{\ket}[1]{\left| #1 \right\rangle}

\newcommand{\ketbra}[2]{\left|#1\middle\rangle\middle\langle#2\right|}

\newcommand{\eg}{\textit{e.g.}\@\xspace}
\newcommand{\ie}{\textit{i.e.}\@\xspace}

\newtheorem{lemma}{Lemma}

\def\be{\begin{equation}}
\def\ee{\end{equation}}

\begin{document}
\title{Inequivalence of entanglement, steering, and Bell nonlocality for general measurements}
\author{Marco T\'ulio Quintino}
\affiliation{Département de Physique Théorique, Université de Genève, 1211 Genève, Switzerland}
\author{Tamás Vértesi}
\affiliation{Institute for Nuclear Research, Hungarian Academy of Sciences,
H-4001 Debrecen, P.O. Box 51, Hungary}
\affiliation{Département de Physique Théorique, Université de Genève, 1211 Genève, Switzerland}
\author{Daniel Cavalcanti}
\affiliation{ICFO-Institut de Ciencies Fotoniques, Mediterranean
Technology Park, 08860 Castelldefels (Barcelona), Spain}
\author{Remigiusz Augusiak}
\affiliation{ICFO-Institut de Ciencies Fotoniques, Mediterranean
Technology Park, 08860 Castelldefels (Barcelona), Spain}
\author{Maciej Demianowicz}
\affiliation{ICFO-Institut de Ciencies Fotoniques, Mediterranean
Technology Park, 08860 Castelldefels (Barcelona), Spain}
\author{Antonio Ac\'in}
\affiliation{ICFO-Institut de Ciencies Fotoniques, Mediterranean
Technology Park, 08860 Castelldefels (Barcelona), Spain}
\affiliation{ICREA-Instituci\'o Catalana de Recerca i Estudis Avan\c cats, Lluis Companys 23, 08010 Barcelona, Spain}
\author{Nicolas Brunner}
\affiliation{Département de Physique Théorique, Université de Genève, 1211 Genève, Switzerland}


\date{\today}  

\begin{abstract}
Einstein-Podolsky-Rosen steering is a form of inseparability in quantum theory commonly acknowledged to be intermediate between entanglement and Bell nonlocality. However, this statement has so far only been proven for a restricted class of measurements, namely projective measurements. Here we prove that entanglement, one-way steering, two-way steering and nonlocality are genuinely different considering general measurements, \ie single round positive-operator-valued-measures. Finally, we show that the use of sequences of measurements is relevant for steering tests, as they can be used to reveal ``hidden steering''.
\end{abstract}

\maketitle

\section{Introduction}
The phenomenon of Einstein-Podolsky-Rosen (EPR) steering, first discussed by Schrödinger \cite{schrodinger36}, represents one form of nonlocality in quantum theory. Consider two distant observers sharing an entangled state. By performing a measurement on his system one observer can remotely steer the state of the system held by the other observer. Often discussed in the context of continuous variable quantum systems \cite{reid89,reid08}, EPR steering was put on firm grounds by Wiseman, Doherty and Jones \cite{wiseman07} who formalized the effect for arbitrary systems. 

Recently a growing interest has been devoted to the notion of steering. Methods for the detection \cite{cavalcanti09,schneeloch13,marciniak14,zukowski14} and quantification \cite{skrzypczyk14,pusey13,kogias14,aolita14} of steering were developed. Experimentally, loophole-free demonstrations of steering have been reported \cite{wittmann11}. Steering is also relevant in the context of quantum information processing \cite{branciard11,law14,piani14}. From a more fundamental viewpoint, steering was shown to be related to the incompatibility of quantum measurements \cite{uola14,quintino14}, and to be able to detect bound entanglement \cite{moroder14}.

Despite these advances it remains an open question whether steering is a form of quantum nonlocality that is inequivalent to entanglement or to Bell nonlocality \cite{bell64,brunner_review}. More precisely, it is unclear if there exist entangled states which are useless for steering, and whether there exist states useful for steering which cannot lead to Bell inequality violation. Although this inequivalence was shown to hold for the particular case of projective measurements \cite{wiseman07,almeida07,saunders10,sania14}, it may not persist if general measurements, not necessarily projective, are taken into account. Note that general measurements were shown to be useful in the context of nonlocality, where they can be used to increase the amount of violation of certain Bell inequalities \cite{vertesi10}. 

\begin{figure}[b]
\includegraphics[width=0.9\columnwidth]{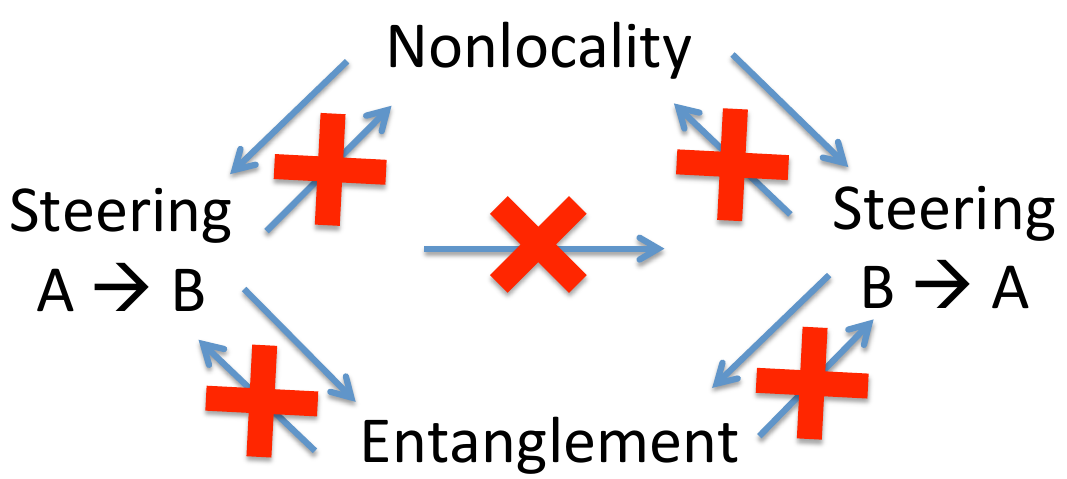} 
\caption{Relations between entanglement, one-way steering, two-way steering and nonlocality. Nonlocality implies two-way steerability (i.e from Alice to Bob and from Bob to Alice), while one-way steerability (\ie from Alice to Bob or from Bob to Alice) implies entanglement. Here we show that the converse relations do not hold, considering arbitrary POVMs (red crosses).}
\end{figure}

Let us also notice that contrary to entanglement and nonlocality, steering features a fundamental asymmetry in the sense that in a steering test the observers play a different role. It is then conceivable that there exist entangled states which are only one-way steerable in the sense that, say, Alice can steer Bob, but Bob cannot steer Alice.  
First investigated in the context of Gaussian systems \cite{midgley09,handchen12}, the effect of one-way steerability was demonstrated for simple two-qubit states, but only for projective measurements \cite{bowles14}. It is again open whether the phenomenon of one-way steering can be observed when general measurements are considered.

Here we show that entanglement, one-way steering, two-way steering, and Bell nonlocality are genuinely different. Specifically, considering here general measurements, we prove the existence of (i) entangled states that cannot lead to steering, (ii) states that can lead to steering but not to Bell nonlocality, and (iii) states which are one-way steerable but not two-way steerable. For each case, we provide a general method for constructing the corresponding states, and discuss explicitly simple examples.

Finally, we also consider the use of sequences of measurements in steering tests, and uncover a phenomenon of ``hidden steering''---by analogy to hidden nonlocality \cite{popescu95,hirsch13}---whereby steering can be activated using local pre-processing.

\section{Scenario} We consider two distant observers, Alice and Bob, performing local measurements on a shared entangled quantum state $\rho$. The measurements are described by POVMs $\{M_{a\vert x}\}$ and $\{M_{b\vert y}\}$ (where $M_{a\vert x} \geq 0$ and $\sum_a M_{a\vert x} =\mathbb{I}$, and similarly for $M_{b\vert y}$), where $x$ and $y$ denote the choice of the measurements and $a$ and $b$ their corresponding outputs. The corresponding probability distributions are given by 
\begin{equation} \label{quantum_prob}
p(ab\vert xy) = \tr(\rho M_{a\vert x}\otimes M_{b\vert y}) .
\end{equation}
The above distribution is termed Bell local when it admits a decomposition of the form
\begin{equation} \label{local}
p(ab\vert xy) = \int d\lambda \pi(\lambda) p_A(a\vert x,\lambda)p_B(b\vert y,\lambda),
\end{equation}
where $\lambda$ is some classical random variable, distributed according to density $\pi(\lambda)$, and $p_A(a\vert x,\lambda)$ and $p_B(b\vert y,\lambda)$ are local response functions. A quantum state $\rho$ is said to be local, or equivalently to admit a local hidden variable (LHV) model, when the statistics of arbitrary local measurements can be reproduced by a distribution of the form \eqref{local} (see Ref. \cite{augusiak_review}). On the contrary, if such a decomposition does not exist, the state is nonlocal and violates a Bell inequality for suitably chosen local measurements \cite{bell64,brunner_review}.

A different notion of nonlocality is that of EPR steering. In a steering test, Bob, who does not trust Alice, wants to verify that $\rho$ is entangled. To this end, he asks Alice to perform measurement $x$ on her subsystem and announce its result $a$.
By doing so, she remotely steers the state of Bob's system to
\begin{equation}\label{assembl}
\sigma_{a\vert x}=\tr_A(M_{a\vert x}\otimes \mathbb{I} \rho),
\end{equation}
where $\tr_A$ denotes the partial trace over Alice's system. Bob's task is now to ensure that the set of conditional states $\{\sigma_{a\vert x}\}$, a so-called assemblage, does not admit a decomposition of the form 

\begin{equation} \label{LHS}
\sigma_{a\vert x}=\int \mathrm{d}\lambda\pi(\lambda) p_A(a\vert x, \lambda)\sigma_\lambda,
\end{equation}
where $\lambda$ is a classical random variable distributed according to density $\pi(\lambda)$, $p_A(a\vert x, \lambda)$ is any possible local response function for Alice, and $\sigma_{\lambda}$ are some quantum states. If the assemblage observed by Bob (\eg via quantum tomography) does admit a decomposition \eqref{LHS}, Bob concludes that Alice could have cheated by using the following strategy: Alice would have sent the (single party, hence unentangled) quantum state $\sigma_\lambda$ to Bob, and announced measurement outcome $a$ according to the response function $p_A(a\vert x, \lambda)$; note that $\lambda$ can be understood here as Alice's choice of strategy. If an entangled state $\rho$ admits a decomposition of the form \eqref{LHS} for all possible measurements, the state is termed unsteerable (or equivalently, it admits a local hidden state (LHS) model). However, if the assemblage $\{\sigma_{a\vert x}\}$ does not admit a decomposition of the form \eqref{LHS}, the state is called steerable. In this case, steering can be detected via violation of a steering inequality \cite{cavalcanti09}. 

Note that LHS models correspond to a special class of LHV models in which the one of the response functions is ``quantum'' (see \cite{augusiak_review}). Hence, any state admitting an LHS model is local, while the converse may not be true. Moreover, due to the asymmetry of the concept of steering, it is in principle possible that there exist entangled states which are only one-way steerable.

The main goal of this work is to fully characterize the relation between entanglement, one-way steering, two-way steering, and nonlocality for general measurements. For the restricted class of projective measurements, all 4 notions are proven to be inequivalent \cite{werner89,wiseman07,bowles14}. However, when considering arbitrary POVMs, it is only known that entanglement is inequivalent to nonlocality \cite{barrett02}. Here we shall see that all 4 notions are inequivalent for POVMs. Specifically, we show by giving explicit examples that (i) there are entangled states that are unsteerable for POVMs, (ii) there exist steerable states admitting an LHV model for POVMs, and (iii) there exist states which are only one-way steerable for POVMs. Finally, we discuss the use of sequences of measurements in steering tests, and uncover the phenomenon of hidden steering.

\subsection{Entanglement vs steering} We first give a method for constructing classes of entangled state admitting an LHS model for POVMs. Specifically, starting from a given entangled state admitting an LHS model for projective measurements, we can construct a different state which (a) admits an LHS model for POVMs, and (b) is entangled. More formally, we can state the following:

\begin{lemma} \label{flag}
Consider an entangled state $\rho$ acting on the Hilbert space $\mathbbm{C}^d\otimes\mathbbm{C}^d$ and admitting an LHS model for projective measurements from Alice to Bob. Then, the state 
\begin{equation} \label{flagstate}
 \rho'=\frac{1}{d+1}  \big[ \rho + d  \, P_{\perp} \otimes \rho_B  \big] 
 \end{equation}
is entangled and admits an LHS model for POVMs, from Alice to Bob. Here, $P_{\perp} $ 
denotes a projector on a subspace that is orthogonal to the support of $\rho_A$, hence $\rho'$ acts on $\mathbbm{C}^{d+1}\otimes\mathbbm{C}^d$. Note that $\rho_{A,B}= \tr_{B,A}(\rho)$ denote the reduced states.
\end{lemma}

\begin{proof} Let us first notice that $\rho'$ is entangled by construction, as one can 
obtain $\rho$ from $\rho'$ by applying a local filter on Alice's side (an operation which cannot produce 
an entangled state from a separable one). Then, the construction of the LHS model for POVMs for $\rho'$ follows directly from the work of Ref. \cite{hirsch13} (see Protocol 2). Note that having an LHS model for binary projective measurements for $\rho$ is in fact enough for the result to hold.
\end{proof}

Applying the above method to known examples of entangled states admitting an LHS model for projective measurements (see Refs. \cite{werner89,wiseman07,almeida07,bowles14,sania14} for examples) allows one to construct entangled states admitting an LHS model for POVMs.

Another example worth mentioning is the Werner states acting on $\mathbbm{C}^d\otimes\mathbbm{C}^d$:
\begin{equation} \label{werner}
\rho_W = \alpha \frac{2 P^{anti}}{d (d-1)} + (1-\alpha) \frac{\mathbb{I}_{d^2}}{d^2}
\end{equation}
where $P^{anti}$ denotes the projector on the antisymmetric subspace, and $\mathbb{I}_{d^2}$ is the identity matrix in dimension $d^2$. It is known that $\rho_W$ is entangled for $\alpha> 1/(d+1)$ \cite{werner89}, and admits an LHV model for POVMs for $\alpha \leq \frac{3d-1}{d+1} (d-1)^{d-1} d^{-d} $ \cite{barrett02}. We point out that the model of Ref. \cite{barrett02} can actually be reformulated as an LHS model; see the supplementary material for details. 
Hence Werner states with $\frac{1}{d+1} < \alpha \leq \frac{3d-1}{d+1} (d-1)^{d-1} d^{-d}$ are also examples of entangled states which cannot be steered using POVMs\footnote{An interesting open question is whether all Werner states admitting an LHS model for projective measurements, \ie with $\alpha\leq \frac{d-1}{d}$, also admit an LHS model for POVMs. See \cite{werner14} for recent progress in this direction.}. Clearly, since Werner states are permutationally invariant, they are two-way unsteerable.

\subsection{Steering vs nonlocality} Before discussing examples of steerable states admitting an LHV model for POVMs, we first derive a useful property of LHS models. Consider a state $\rho$ admitting an LHS model from Alice to Bob (\ie Alice cannot steer Bob). Then, any local probabilistic transformation represented by a trace non-increasing completely positive (CP) map (for instance a local filtering) Bob may apply to his system must leave the global state $\rho$ unsteerable (from Alice to Bob). This fact, already noticed in Refs \cite{uola14,aolita14}, can be formally stated as follows.

\begin{lemma}\label{SLOCC}
	Let $\rho$ be an entangled state, unsteerable from Alice to Bob. For any local operation represented by an arbitrary trace non-increasing CP map $\Lambda$ acting on Bob's system, the final state 
\begin{equation}\label{filtered}
\rho_{F}=\frac{(I\otimes\Lambda)(\rho)}{\tr[(I\otimes\Lambda)(\rho)]},
\end{equation}
where $I$ stands for the identity map, is unsteerable from Alice to Bob. 
\end{lemma}

\begin{proof} 
Since $\rho$ admits an LHS model (from Alice to Bob) we have that any assemblage $\{\sigma_{a\vert x}\}$ generated from $\rho$ admits a decomposition of the form \eqref{LHS}. Using this fact, we will now prove that $\rho_F$ also admits an LHS model (from Alice to Bob).
To this end, let us first notice that an assemblage $\{\widetilde{\sigma}_{a\vert x}\}$ obtained 
from $\rho_F$ by Alice performing a measurement $\{M_{a|x}\}$
is related to $\{\sigma_{a|x}\}$ through
\begin{equation}\label{tilde}
\widetilde{\sigma}_{a\vert x}=\frac{1}{p_F}\Lambda(\sigma_{a\vert x})
%
\end{equation}
where $p_F=\tr[\Lambda(\rho_B)]$
%
%
with $\rho_B=\tr_A\rho$. After inserting Eq. (\ref{LHS})
into Eq. (\ref{tilde}) and then rearranging the terms, one obtains 
\begin{equation}
\widetilde{\sigma}_{a\vert x}=\int\mathrm{d}\lambda\frac{\pi(\lambda)\omega(\lambda)}{p_F}p(a\vert x,\lambda)\widetilde{\sigma}_{\lambda}
%
\end{equation}
where $\omega(\lambda)=\tr[\Lambda(\sigma_{\lambda})]$, and $\widetilde{\sigma}_{\lambda}=\Lambda(\sigma_{\lambda})/\omega(\lambda)$
is a normalized density matrix. To complete the proof, note that 
$\widetilde{\pi}(\lambda)=\pi(\lambda)\omega(\lambda)/p_F$ is a proper probability density. In particular, 
it follows from Eqs. (\ref{assembl}) and (\ref{LHS}) that $\rho_{B}=\sum_a\sigma_{a\vert x}=\int\mathrm{d}\lambda\pi(\lambda)\sigma_\lambda$, and therefore $(1/p_F) \int\mathrm{d}\lambda\pi(\lambda)\omega(\lambda)=1$.
\end{proof}

A relevant corollary of Lemma \ref{SLOCC} is the following. Consider again an entangled state $\rho$ and a local operation on Bob's side. If Alice can steer Bob with the final state $\rho_F$, then the initial state $\rho$ must also be steerable (from Alice to Bob). Moreover, if $\rho_F$ violates a given steering inequality, one can construct another steering inequality which can be violated by the initial state $\rho$ (see the supplementary material for details). A particular example of such local operation $\Lambda$ is the local filtering where $\Lambda(\cdot)=F(\cdot)F^{\dagger}$ with $F$ being any matrix satisfying $F^{\dagger}F\leq \mathbb{I}$.

We are now in position to discuss our examples of steerable states, which are nevertheless local for all POVMs. Consider the class of states of the form
\begin{align}\label{GM}
 \rho_{G} & =  \frac{1}{9} \big[ q \ket{\psi^-}\bra{\psi^-} + (3-q) \frac{\mathbb{I}_2}{2}  \otimes \ketbra{2}{2} \nonumber   \\ & +  2q \ketbra{2}{2} \otimes \frac{\mathbb{I}_2}{2} + (6-2q) \ketbra{22}{22} \big], 
\end{align}	
where $\ket{\psi^-}= (\ket{01}-\ket{10})/\sqrt{2}$ is the two-qubit singlet state, and $\mathbb{I}_2$ denotes the identity in the qubit subspace $\{ \ket{0},\ket{1} \}$. For $0< q \leq 1/2$, these states are proven to be local for POVMs \cite{hirsch13}. 

We will now see that $\rho_G$ is steerable (in both directions) for $0< q \leq 1$. Notice that if Alice applies a local filtering on the qubit subspace $\{ \ket{0},\ket{1} \}$, the filtered state is of the form

\begin{equation}  \label{erasure}
\rho_{G}^F = \alpha \ket{\psi^-}\bra{\psi^-}  + (1-\alpha)  \frac{\mathbb{I}_2}{2}\otimes\ketbra{2}{2}
\end{equation}
with $\alpha= q/3$. Note that states of the form \eqref{erasure} are so-called ``erasure states'' (as they can be obtained by sending half of a singlet state $\ket{\psi^-}$ through an erasure channel). If Bob applies a local filtering on the qubit space, the filtered state is also an erasure state (where the subsystems are swapped), with $\alpha=1/3$. Since the erasure state is steerable (in both directions) for $0< \alpha \leq 1$, it follows from the corollary to Lemma 1 that the state $\rho_G$ with $0< \alpha \leq 1/2$ is two-way steerable but local for all POVMs.

The steerability of the erasure state for any $0< \alpha \leq 1$ deserves a few more explanations. First notice that steerability from Alice to Bob follows again from our corollary: when Bob applies a projection on the qubit subspace $\{ \ket{0},\ket{1} \}$, the filtered state is simply the pure singlet state $\ket{\psi^-}$, which is clearly steerable. Steering from Bob to Alice can be demonstrated by considering an explicit family of steering inequalities \cite{skrzypczyk15}.

\subsection{One-way steering for POVMs} We give here a simple technique for constructing states that are unsteerable from Alice to Bob for arbitrary POVMs, but steerable from Bob to Alice. Notice that by construction such states are local for POVMs. The technique will then be illustrated by an example. 

The idea of the method is to start from a state that is one-way steerable (that is steerable from Bob to Alice, but not from Alice to Bob) for projective measurements (examples were provided in \cite{bowles14}), and then construct a state that is one-way steerable for POVMs. More formally we have that:

\begin{lemma}  \label{POVM_EPR}
Let $\rho$ be a quantum state acting on the Hilbert space $\mathbbm{C}^d\otimes\mathbbm{C}^d$ such that Alice cannot steer Bob with projective measurements but Bob can steer Alice. Then, the state $\rho'$, defined as in Eq. \eqref{flagstate}, is such that Alice cannot steer Bob using arbitrary POVMs, but Bob can steer Alice. 
\end{lemma}

\begin{proof}
First, notice that since $\rho$ admits an LHS model for projective measurements from Alice to Bob, then the state $\rho'$ admits an LHS model for POVMs from Alice to Bob, which follows from Lemma \ref{flag}. Second, there exists a local operation that allows Alice to map the state $\rho'$ to $\rho$, which, by assumption, is steerable from Bob to Alice. Hence, it follows from Lemma \ref{SLOCC} that $\rho'$ is steerable from Bob to Alice.
\end{proof}

To provide an explicit example of a state featuring one-way steering for POVMs, we first consider the state
\begin{equation}\label{oneway_state}
\rho_{1W}= \frac{1}{2}  \big[ \ket{\psi^-}\bra{\psi^-} +  \frac{3}{5}\ketbra{1}{1}\otimes \frac{\mathbb{I}_2}{2} + \frac{2}{5} \frac{\mathbb{I}_2}{2}\otimes \ketbra{0}{0} \big]
\end{equation}
which cannot be steered from Alice to Bob, considering projective measurements, but Bob can steer Alice using 13 well-chosen measurements \cite{bowles14}. By applying the above Lemma, we can construct the state 

\begin{equation}
\rho'_{1W}= \frac{1}{3} \rho_{1W} + \frac{2}{3} \ketbra{2}{2}\otimes \rho_{1W,B}
\end{equation}
which is one-way steerable for POVMs.

\section{Sequences of measurements} We have shown above that entanglement, one-way steering, two-way steering, and nonlocality are genuinely different when considering general measurements. A natural question is to see whether the relations between these notions of quantum nonlocality change when states are subjected to sequences of measurements. Sequences of measurements are relevant in the context of Bell nonlocality as they allow one to detect the nonlocal properties of quantum states that have a LHV model for general measurements, a phenomenon known as ``hidden nonlocality'' \cite{popescu95,hirsch13}. Below, we show that a similar effect is possible for steering. Specifically, we demonstrate that a state which admits an LHS model for POVMs (in both directions), can lead to steering when Alice and Bob can perform a sequence of measurements. 

Consider the Werner state \eqref{werner} with $\alpha = (d-1)/d$, denoted $\tilde{\rho}_W$, which admits an LHS model for projective measurements \cite{werner89}. Now, using Lemma \ref{flag} twice in a row (first performing the extension on Alice's side, and then on Bob's side), we obtain the following state
\begin{align}
\rho_{HS} =& \frac{1}{d^2} \big[   \tilde{\rho}_W  + d ( P_{\perp} \otimes \frac{\mathbb{I}_d}{d} + \frac{\mathbb{I}_d}{d} \otimes P_{\perp})  + d^2 P_{\perp} \otimes P_{\perp} \big],
\end{align}
where $P_{\perp}$ is a projector on a subspace orthogonal to the support of the reduced states of Alice and Bob. By construction $\rho_{HS} $ admits an LHS model for POVMs from Alice to Bob and from Bob to Alice. Consider now applying the local filters $F_A = F_B = \ket{0}\bra{0} + \ket{1}\bra{1}$, \ie projections onto a qubit subspace, on both Alice's and Bob's side. The resulting state is given 
\begin{equation}
\frac{(F_A \otimes F_B) \rho_{HS} (F^{\dagger}_A \otimes F^{\dagger}_B)}{ \tr[(F_A \otimes F_B) \rho_{HS} (F^{\dagger}_A \otimes F^{\dagger}_B)]}
 = \frac{1}{1+\frac{2}{d}} \big[ \ket{\psi^-}\bra{\psi^-} + \frac{2}{d} \frac{\mathbb{I}_4}{4} \big]
\end{equation}
which is steerable (in both directions) for any $d\geq 3$ \cite{wiseman07}. Hence, the state $\rho_{HS}$ has hidden steering. 

Note that the notion of hidden steering is intimately related to one-way steerability. In fact, novel examples of states with one-way steering can be easily constructed from the above example of hidden steering (see supplementary material).

\section{Discussion} We have shown that entanglement, steering, and nonlocality are inequivalent when general measurements are considered. The natural question is now to see how these notions relate to each other when sequences of measurements are allowed. While we are not in position to give a final answer, we could nevertheless show that sequences of measurements are relevant for demonstrating steering. More generally, it is in fact not known whether entanglement and nonlocality are strictly inequivalent in this case, despite recent progress \cite{masanes08,vertesi14}.

\textit{Acknowledgements.} We thank Joe Bowles, Flavien Hirsch, and Paul Skrzypczyk for discussions. We acknowledge financial support from the Swiss National Science Foundation (grant PP00P2\_138917 and Starting Grant DIAQ), SEFRI (COST action MP1006), the 
J\'anos Bolyai Programme, the OTKA (K111734), the ERC CoG QITBOX and AdG OSYRIS, the EU project SIQS, the Spanish project FOQUS, the John Templeton Foundation and the Generalitat de Catalunya. R. A is supported by the Spanish MINECO through the Juan de la Cierva scholarship and D. C. is supported by the Beatriu de Pin\'os fellowship (BP-DGR 2013).



\appendix

\section{Converting Barrett's LHV model for Werner states into an LHS model}

We show that the Werner states discussed in the main text admit an LHS model for arbitrary non-sequential POVMs. This is done by showing that Barrett's model \cite{barrett02} for simulating POVMs on these Werner states can be straightforwardly transformed into an LHS model. 

Without going into full details about the model, we recall that the shared variable $\lambda$ can be viewed as a quantum state of dimension $d$: $\ket{\lambda}$. Alice's response function given by 
\begin{equation}
p_A(a\vert M_{a|x},\lambda) = \frac{\alpha_{a|x} }{d-1}(1- \tr[\ketbra{\lambda}{\lambda}  P_{a|x}]),
\end{equation}
where $P_{a|x}$ is a rank one projector defined by $M_{a|x}=\alpha_{a|x} P_{a|x}$ the POVM elements. By noticing that

\begin{align}
\frac{\alpha_{a|x} }{d-1}(1- \tr[\ketbra{\lambda}{\lambda} P_{a|x}]) &= \frac{\alpha_{a|x} }{d-1}\tr[(\mathbb{I}_d- \ketbra{\lambda}{\lambda}) P_{a|x}]  \nonumber \\
&= \tr\left( \frac{\mathbb{I}_d- \ketbra{\lambda}{\lambda} }{d-1} \alpha_{a|x}P_{a|x}\right) \nonumber \\
&= \tr\left( \frac{\mathbb{I}_d- \ketbra{\lambda}{\lambda} }{d-1} M_{a|x}\right), 
\end{align}	 
one can define a new shared variable, as the quantum state $\sigma_{\lambda}=\frac{1 }{d-1}(\mathbb{I}_d- \ketbra{\lambda}{\lambda}) $, in order to transform the initial LHV model into an LHS model.

\section{Constructing families of steering inequalities}

Here we provide a simple method for constructing a steering inequality 
violated by a state $\rho$, starting from a steering inequality violated by the transformed state 
\begin{equation}\label{Ginevra}
\rho_{F}=\frac{1}{p_F}(I\otimes\Lambda)(\rho)
\end{equation}
where $\Lambda$ is any trace non-increasing completely positive map and 
$p_F=\tr[(I\otimes \Lambda)(\rho)]$ is the probability that a quantum operation represented by this map 
has been successfully implemented. For this purpose, let us first prove the following Lemma.

\begin{lemma}\label{SLOCC2}

Consider the steering inequality 
	\begin{equation}\label{Genf}
	\tr\sum_{a,x} \Gamma_{a\vert x} \sigma_{a\vert x}^{uns} \leq 0,
	\end{equation} 
	where $\Gamma_{a\vert x}$ are the operators characterizing the inequality (see \cite{skrzypczyk14} for details), and $\{\sigma_{a\vert x}^{uns}\}$ an arbitrary unsteerable assemblage; notice that by redefining the operators $\Gamma_{a\vert x}$ one can always set the bound of any steering inequality to zero. Then, for any completely positive map 
	$\Lambda$,
	\begin{equation} \label{Geneva}
	\tr\sum_{a,x} [\Lambda^{\dagger}(\Gamma_{a\vert x}) \sigma_{a\vert x}^{uns}] \leq 0,
	\end{equation} 
is also a steering inequality. Here by $\Lambda^{\dagger}$ we denote a dual map of $\Lambda$\footnote{A dual map $\Lambda^{\dagger}$ of some linear map $\Lambda$ is one that satisfies $\tr[X\Lambda(Y)]=\tr[\Lambda^{\dagger}(X)Y]$ for any pair of matrices $X,Y$.}.
\end{lemma}

\begin{proof} First, using the definition of a dual map, one can rewrite 
the left-hand side of inequality \eqref{Geneva} as
\begin{equation}
\tr\sum_{a,x} [\Lambda^{\dagger}(\Gamma_{a\vert x}) \sigma_{a\vert x}^{uns}]=
p_F\tr\sum_{a,x} \left[\Gamma_{a\vert x} \frac{1}{p_F}\Lambda(\sigma_{a\vert x}^{uns})\right],
\end{equation}
where $p_F=\tr[\Lambda(\rho_B)]$ with $\rho_B$ denoting the second subsystem of an unsteerable state $\rho$ giving
$\{\sigma_{a|x}^{\mathrm{uns}}\}$. Then, as shown in the proof of Lemma 2, for any unsteerable assemblage $\{\sigma_{a|x}^{\mathrm{uns}}\}$ the operators $\widetilde{\sigma}_{a\vert x}^{\mathrm{uns}}=\Lambda(\sigma_{a\vert x}^{\mathrm{uns}})/p_F$ form another unsteerable assemblage that corresponds to the state $\rho_F$ given in Eq. \eqref{Ginevra}.

This means that inequality (\ref{Genf}) is satisfied for $\{\widetilde{\sigma}_{a\vert x}^{\mathrm{uns}}\}$, which together with the fact that $p_F\geq 0$, implies (\ref{Geneva}). 
This completes the proof.
\end{proof}
Now, let us see how this method works in practice. Assume that a state $\rho_{F}$ given by 
Eq. (\ref{Ginevra}) violates some steering inequality (\ref{Genf}) by the amount
	\begin{equation}
	\beta_{F}= \tr\sum_{a,x} \Gamma_{a\vert x} \tr_A[M_{a|x}\otimes \mathbb{I} \rho_F]  >0.
	\end{equation} 
	Then, it clearly follows that $\rho$ violates the following inequality
	\begin{equation} \tr\sum_{a,x} [\Lambda^{\dagger}(\Gamma_{a\vert x}) \sigma_{a\vert x}^{uns}] \leq 0,
	\end{equation} 
which, as proven in Lemma \ref{SLOCC2}, is a proper steering inequality,
and the amount of violation is 
\begin{equation}
\tr\sum_{a,x} \left[\Lambda^{\dagger}(\Gamma_{a\vert x}) \tr_A(M_{a|x}\otimes \mathbb{I} \rho)\right] = \beta_{F} p_F>0.
\end{equation}

\section{From hidden steering to one-way steering}

In this section we present a general technique for constructing a state with one-way steering, starting from a state featuring hidden steering. Before stating the general result we show how to construct a novel example of a one-way steerable state starting from a Werner state. 
	
As discussed in the main text, the local model presented by Werner \cite{werner89} is an LHS model. Hence if one party (say Bob) projects his subystem onto the qubit subspace $\{\ket{0},\ket{1}\}$, \ie applying the filter $F_{01}= \proj{0}+\proj{1}$, the filtered state 
\begin{equation}
	\rho_W^{F}= \frac{1}{N} (\mathbb{I}_d \otimes  F_{01}) \rho_W   (\mathbb{I}_d \otimes  F_{01}),
\end{equation}
where $N= \tr[(\mathbb{I}_d \otimes  F_{01}) \rho_W   (\mathbb{I}_d \otimes  F_{01})]$, is unsteerable from Alice to Bob. However, Bob can steer Alice whenever $d\geq 3$. This follows from the fact that, if Alice now also projects her subsystem onto the qubit subspace $\{\ket{0},\ket{1}\}$, the final state is a two qubit Werner state with visibility greater than $1/2$, which is steerable. 

More generally, we have the following result.

\begin{lemma} \label{HS=oneway}
Consider a state $\rho$ such that Alice cannot steer Bob, but Bob can steer Alice with the filtered state 
\begin{equation}
\rho'= \frac{(F_A\otimes F_B) \rho (F_A^{\dagger}\otimes F_B^{\dagger})}{\tr[(F_A\otimes F_B) \rho (F_A^{\dagger}\otimes F_B^{\dagger})]}.
\end{equation}
Then, the state 
\begin{equation}
\rho''= \frac{(I \otimes F_B) \rho (I \otimes F_B^{\dagger})}{\tr[(I \otimes F_B) \rho (I \otimes F_B^{\dagger})]}
\end{equation}
is one-way steerable: Alice cannot steer Bob, but Bob can steer Alice.
\end{lemma}

\begin{proof}
It follows directly from lemma 2 that Alice cannot steer Bob with the state $\rho''$. It also follows from lemma 2 that if Bob can steer Alice with $\rho'$, he can steer Alice with $\rho''$.
\end{proof}

Finally, note that if $\rho$ has hidden steering for projective measurements, the state $\rho''$ is one-way steerable for projective measurements. Whereas, if $\rho$ has hidden steering for POVMs, the state $\rho''$ is one-way steerable for POVMs.


\providecommand{\href}[2]{#2}\begingroup\raggedright\endgroup


\end{document}